\documentclass[12pt]{article}
\usepackage{amsthm,amsmath,amssymb}
\usepackage{epsfig,graphicx,graphics,color,xcolor}
\usepackage{appendix}
\usepackage{enumerate}

\usepackage{tikz,color,xcolor}
\usetikzlibrary{positioning}

\usepackage{float}
\usepackage{amssymb,amsmath,amsthm,enumitem} 
\usepackage{subcaption,booktabs}
\usepackage{hyperref}
\usepackage{tikz}
\usepackage{tabularray} 
\usepackage[export]{adjustbox}
\usepackage{authblk}
\usepackage[export]{adjustbox}
\UseTblrLibrary{booktabs}
\usetikzlibrary{positioning,shapes.misc,shapes.geometric, graphs, calc}

\newcommand{\gs}{\sigma}

\newcommand{\gO}{\Omega}

\newtheorem{theorem}{Theorem}

\newtheorem{corollary}{Corollary}

\newtheorem{lemma}{Lemma}

\newtheorem{proposition}{Proposition}
\newtheorem{remark}{Remark}

\begin{document}


\title{On Nash Equilibria in 
Play-Once and Terminal Deterministic Graphical Games}

\author{
Endre Boros\footnote
{MSIS \& RUTCOR, Business School, Rutgers University,
100 Rockafellar Road, Piscataway NJ  08854;
\textit{endre.boros@rutgers.edu}}
\and
Vladimir Gurvich
\footnote
{National Research University, Higher School of Economics, Moscow, Russia; \\
RUTCOR, Rutgers University,
100 Rockafellar Road, Piscataway NJ  08854;
\textit{vladimir.gurvich@gmail.com}}
\and Kazuhisa Makino
\footnote{Research Institute for Mathematical Sciences, 
Kyoto University, Kyoto, Japan;
\textit{makino@kurims.kyoto-u.ac.jp}}
}
\maketitle


\begin{abstract}
We consider finite $n$-person 
deterministic graphical games
and study the existence of pure stationary Nash-equilibrium  
in such games. 
We assume that all infinite plays are equivalent 
and form a unique outcome, while each terminal position 
is a separate outcome.
It is known that for $n=2$ such a game always has a Nash equilibrium, 
while that may not be true for $n > 2$.
A game is called {\em play-once} if 
each player controls a unique position and 
{\em terminal}  if any terminal outcome is better 
than the infinite one for each player. 
We prove in this paper that play-once games have Nash equilibria. 
We also show that terminal games have Nash equilibria if they have at most three terminals. 
\newline
{\bf Keywords:}
$n$-person deterministic graphical game,
play-once game, terminal game, 
Nash equilibrium, Nash-solvability,
pure stationary strategy, digraph, directed cycle. 
\end{abstract}

\section{Introduction}

We consider a family of games with perfect information and 
without moves of chance, played on finite directed graphs.  
Given a directed graph $G = (V, E)$, the sets $V$ and $E$ 
represent positions and moves of the players, respectively.
We are also given a partition 
$V = V_1 \cup \dots \cup V_n \cup T$
and an initial position $v_0 \in V \setminus T$, where 
$V_i$  is the set of positions controlled by
player $i \in I = \{1, \dots, n\}$ and 
$T$ is the set of terminals, that is, 
the vertices of out-degree 0. 

In the game each player $i\in I$ chooses a 
move from each position $v \in V_i$. 
Such arcs form a subgraph, called a \textit{situation}, in which there is a unique walk starting at the initial position. That walk may end in a directed cycle, or at a terminal position. Accordingly, the outcome of the game is defined as $\infty$ if the walk ended in a cycle, or the terminal position $a\in T$ at which the walk ended. 
Players have given preferences over the set of outcomes, $A=T\cup \{\infty\}$. 
Given a situation, a player $i\in I$ may be able to improve the outcome according to its own preferences by choosing different outgoing arcs from the positions $v\in V_i$. 
A situation is called a \textit{Nash equilibrium} (NE) if none of the players have such an improving deviation from it. 

This family of games was introduced by Washburn \cite{Was90} 
in 1990 for the case of $n=2$ players and 
opposite preferences (zero-sum games). 
Furthermore, every non-terminal outcome $\infty$ was treated as a draw. 
For this reason, these games were referred to as ``chess-like" 
in \cite{BEGM12,BG09,BGOR14,GO14}. 
Several other publications consider 
a natural extension of these games for the case $n \geq 2$  
\cite{AGH10,AHMS12,BEGM12,BG03,BGMW11,DS02,GO14}. 
Washburn \cite{Was90} called his  
games \textit{deterministic graphical} (DG)  
and we use the same name for the above more general family.

It is known that 2-person DG games have NE \cite{BG03}. 
This result is based on a more general theory of 
Nash-solvability of tight game forms \cite{Gur75,Gur89},  
while the zero-sum case was studied earlier
by Edmonds and Fulkerson \cite{EF70}, see also \cite{Gur73,Was90} 
and some later works \cite{Gur21,GN21,GN21A,GN21B,GN22}. 
The existence of NE (in pure strategies)   
for several other classes of games was studied in  
\cite{ARV09,AGH10,AHMS12,BG09,BGMW11,Gur17,Gur18,Gur21a,BEGV24,GK18,HL97,
HR04,KMMM90,Mil96,Mil96_2,MS96,Ros73}.

However, an $n$-person DG game may have no NE when $n>2$;  
see examples in \cite{GO14} for $n=4$ and \cite{BGMOV18} for $n=3$.  

A digraph $G = (V,E)$ is called symmetric 
if $(u,w) \in E$  whenever $(w,u) \in E$ 
and position $w$ is not a terminal. 
It is also known that 
a DG game with 
a symmetric digraph has a NE \cite{BFGV23}. 

In this paper we focus on two special families of DG games. 
In the so called \textit{play-once} games we have $|V_i|=1$ for all $i\in I$, 
while in \textit{terminal} games all players dislike the outcome $\infty$ and 
prefer to it any 
terminal outcome. 

The structure of play-once and terminal DG games was studied in \cite{BG03}, 
where it was shown that such games have NE. 
We conjecture that, in fact, either of these properties imply the existence of a NE. 
While we can only prove a partial result about the terminal case, 
we can fully verify the conjecture for the play-once games.

\begin{theorem}\label{t-3T}
    A terminal DG game has a NE if $|T|\leq 3$. 
    Furthermore, it has a terminal NE if in addition some terminal can be reached by a directed path from the initial position. 
\end{theorem}

\begin{theorem}\label{t-playonce}
    A play-once DG game has a NE.
\end{theorem}

In Section \ref{s2} we provide further terminology and notation, 
in Section \ref{s3} we show some simple claims, 
in Section \ref{s4} we analyze terminal games and prove Theorem \ref{t-3T}, 
and in Section \ref{s5} we consider play-once games and 
prove Theorem \ref{t-playonce}. Finally, in Section \ref{s6} we pose some open questions.

\section{Some Concepts, Definitions, and Notation}
\label{s2} 

\noindent\textbf{Graphs modeling  DG games}: 
We denote by $I=\{1, \dots , n\}$  the set of players.  
Let $G=(V,E)$ be a directed graph, 
where vertices and directed edges are respectively called {\em positions} and {\em moves}. 
The vertex set $V$ is partitioned by 
$V = T \cup \bigcup_{i \in I}V_i$, 
where $V_i$ i s the set of positions controlled
by player $i \in I$, 
while $T$ is the set of terminals, i.e., positions with no moves from them. 
For a player $i\in I$ we denote by $E_i\subseteq E$ the moves \emph{controlled by player} $i$, that is, $E_i=\{(v,w)\in E\mid v\in V_i\}$.
For a position $u\in V\setminus T$ we denote by $i(u)\in I$ the player who controls that position, that is $u\in V_{i(u)}$.
We have a designated position $v_0 \in V\setminus T$, called the \emph{initial position}. 

Given a subset $F \subseteq E$ and a vertex $v \in V$ 
we denote by $N_F^+(v)$ the out-neighborhood of $v$ 
and by $d_F^+(v)=|N_F^+(v)|$  its cardinality, so-called out-degree.   
By definition, we have $d_E^+(t)=0$ for all terminals $t\in T$ and  $d_E^+(v)>0$ for all positions $v\in V\setminus T$. 
For a subset $X\subseteq V$ of the positions, 
we denote by $G[X]=(X,E[X])$ the subgraph of $G$ induced by $X$. 


\medskip

\noindent\textbf{Strategies and situations of DG games}: 
A \emph{strategy} of a player $i\in I$ is a mapping $\gs_i:V_i\to E_i$ 
that assigns the move $\gs_i(v)\in \{(v,w) \in E_i\}$ 
for each position  $v\in V_i$. 
Note that $\gs_i(v)$ depends only on the position $v \in V_i$, and hence, the strategy is called  \emph{positional} or
\emph{stationary} in the literature. 
We denote by $\Sigma_i$ the set of strategies of player $i\in I$, and note that
$|\Sigma_i| = \prod_{v\in V_i} d^+_E(v)$.

For simplicity, we also view $\gs_i$ as a subset of the moves of $G$: $\gs_i=\{(v,\gs_i(v))\mid v\in V_i\}\subseteq E_i$ for $i\in I$. In other words, a strategy $\gs_i\subseteq E$ satisfies  $d^+_{\gs_i}(v)=1$ for all $v\in V_i$ and $0$ for all $v\in V\setminus V_i$. 
A \emph{situation} is a collection $\gs=(\gs_i\in \Sigma_i \mid i\in I)$ of strategies, one for each player. We can also view a situation as a subset of moves $\gs\subseteq E$ such that $d^+_{\gs}(v)=1$ for all positions $v\in V\setminus T$ and $0$ for all terminals $v\in T$. 

\medskip
\noindent\textbf{Mechanics and outcomes of DG games}: 
Each player chooses a strategy. 
In the corresponding situation $\gs$, for any position $v \in V \setminus T$,  
there exists a unique walk $P(\gs,v)$ leaving  $v$. 
Such a walk terminates either in a terminal, 
or in a cycle that is then repeated infinitely many times. 
We define
\[
g(\gs,v) ~=~ 
\begin{cases}
	a & \text{ if the path } P(\gs,v) \text{ is finite that terminates }  a\in T, \\
	\infty & \text{ if the path } P(\gs,v) \text{ is infinite.}
\end{cases}
\]
We call $P(\gs)=P(\gs, v_0)$ the \emph{play} of the situation and denote by $g(\gs)=g(\gs, v_0)$ the corresponding outcome, which we call the \emph{outcome} of the situation. 
Accordingly, in \emph{DG games} 
we define the set of outcomes 
as the set of terminals and one non-terminal outcome, 
$A = T \cup \{\infty\}$, and call the play $P(\gs)$ \emph{terminal or finite} 
if $g(\gs)\in T$, and \emph{infinite} otherwise. 

Given a situation $\gs$, for a player $i\in I$ we denote by 
$\gs_{-i}=\bigcup_{j\in I\setminus\{i\}}\gs_j$ 
the set of moves of $\gs$ controlled by players from $I \setminus \{i\}$. 
Thus, for every $i\in I$ we have $\gs=\gs_{-i}\cup\gs_i$. 

\medskip

\noindent\textbf{Preferences and purpose of  DG games}: 
Each player $i$ has a preference over the outcomes  
that is described by a linear order $\succ_i$ over the set  $A$. 
We denote by $\succ$ 
 the collection $\{\succ_i \,\,\,\mid i\in I\}$ of preference orders, and denote by $\Gamma=\Gamma(G=(V=T\cup \bigcup_{i \in I}V_i,E),v_0\in V\setminus T,\succ)$ the corresponding game. 

 Furthermore, for a subset $X\subseteq V$ such that $T \cup \{v_0\} \subseteq X$ we denote by $\Gamma[X]=(G[X]=(X,E[X]), v_0,\succ)$ the subgame of $\Gamma$ induced by the set $X$. 
If ambiguity may arise, we write $\Sigma_i(\Gamma)$ for the set of strategies of player $i\in I$ in game $\Gamma$.

\medskip

\noindent\textbf{Nash equilibria in DG games}: 
We say that a situation $\gs$ is a \emph{Nash equilibrium} (or NE, in short), if
\[
g(\gs) ~\succeq_i ~ g(\gs_{-i}\cup\gs'_i)
\]
for all players $i\in I$ and for all strategies $\gs'_i\in \Sigma_i$. In other words, a situation $\gs$ is a NE if none of the players can strictly improve (according to his/her own preferences) if all the other players stick to their chosen strategies. 
We say that a NE $\gs$ is {\em terminal} if  $g(\gs)\in T$ and {\em non-terminal} otherwise. 

Since we consider NE of DG games,  
we restrict our attention to linear orders as preferences, rather than orders with ties or payoffs.
Obviously, tie-breaking might destroy NE, but cannot create one. 
In contrast, merging outcomes might only create a NE  but cannot destroy one. 

\medskip

\noindent\textbf{Backward Induction}: 
If the digraph $G$ is acyclic then 
a special NE can be obtained in a DG game  
$G, D, \succ)$  by the well-known 
Backward Induction (BI). 
This procedure was introduced in  1953 by Gale and Kuhn  \cite{Gal53;Kuh53}. 
More properties of the BI NE can be found in \cite{Gur17,Gur18}.

\medskip

\section{Some simple proofs and examples}
\label{s3}

Let us first consider an easy special case.

\begin{lemma}\label{l-notreachable}
Let $\Gamma=(G=(V,E),v_0,\succ)$ be a DG game in which the set of terminals is not reachable from the initial position. Then $\Gamma$ has a non-terminal NE. 
\end{lemma}

\begin{proof}
    Let us denote by $S$ the set of positions that can be reached from $v_0$ by a directed path. According to our assumption, we have $S\cap T=\emptyset$. 

    Let us  define a situation $\sigma$ by choosing a non-terminating move from each position $v\in S$, and an arbitrary move from positions $v\in V\setminus (S\cup T)$. 
    Note that a  terminal DG game has at least one move from all positions in $V\setminus T$, by definition.

    Then the play $P(\sigma)$ is not leaving the set $S$, and hence must terminate in a cycle within $S$. Furthermore, all players on this play belong to $S$, and hence none of them can switch to a strategy that would leave $S$. Consequently, $\sigma$ is a NE with $g(\sigma)=\infty$, as claimed.
\end{proof}

\medskip

Let us show next a simple proof for the result that we recalled earlier, that is that if a DG game is both terminal and play-once, then it has a NE. 
This statement follows from the results of \cite{BG03}. 
Here we give a simpler direct proof. 

\begin{proposition}
A terminal and play-once DG game has an NE.
\end{proposition}

\begin{proof}
Assume that $\Gamma=(G=(V,E),v_0,\succ)$ is a terminal and play-once DG game in which the a terminal position is reachable from the initial position. Due to Lemma \ref{l-notreachable} we can assume this, since otherwise $\Gamma$ has a non-terminal NE.

We can also assume that all non-terminal positions of $\Gamma$ have a non-terminating move, since otherwise we could fix a best terminating move for such a position and reduce the game to a smaller, equivalent one. 

Let us then consider a shortest $v_0\to T$ path $P$ (that has the minimum  number of arcs) such that the last position $v$ on $P$ before $T$ chooses a best terminal move $(v,a)$ for $a\in N^+(v)\cap T$ according to his preferences.  We can then define a situation $\sigma$ by using the moves along $P$ for positions on $P$, and choosing a non-terminating move for all other non-terminal positions. 

We claim that $\sigma$ is a terminal NE. Note that a player $u\neq v$ on the play $P(\sigma)=P$ cannot reach $T$ by a single move, since then $P$ would not be a shortest path. Thus, all deviations from $\sigma$ by players different from $v$ on the play yield $\infty$, while a deviation by $v$ may yield another terminal that $v$ does not prefer to $g(\sigma)$. Hence $\sigma$ is a NE, as claimed.
\end{proof}

Let us recall a simple example from \cite{BGMOV18} for a NE-free $3$-terminal DG game, which is not terminal. 

\begin{figure}[ht]
	\begin{center}
		\begin{tikzpicture}[scale=0.25,fill opacity=0.5]
			\node[rectangle,draw,fill=blue,opacity=0.5] (v1) at (0,10) {$1$};
			\node[circle,draw,fill=red,opacity=0.5] (v2) at (10,10) {$2$};
                \node[circle,draw,fill=red,opacity=0.5] (v3) at (10,0) {$2$};
			\node[circle,draw,fill=green,opacity=0.5] (v4) at (0,0) {$3$};
			
			\node[circle,draw,fill=white] (a) at (15,10) {$a$};
			\node[circle,draw,fill=white] (b) at (15,-5) {$b$};
			\node[circle,draw,fill=white] (c) at (0,-5) {$c$};
			\draw[thick,->] (v1) to (v2);
			\draw[thick,->] (v2) to (a);
			\draw[thick,->] (v2) to (v3);
			\draw[thick,->] (v3) to (b);
			\draw[thick,->] (v3) to[bend left] (v4);
                \draw[thick,->] (v4) to[bend left] (v3);
                \draw[thick,->] (v1) to (v4);                
                \draw[thick,->] (v4) to (c);

                \node[align=right] at (25,6) {$b\succ_1 \infty\succ_1 a\succ_1 c$};
                \node[align=right] at (25,3) {$c\succ_2 a\succ_2 b\succ_2 \infty$};
                \node[align=right] at (25,0) {$a\succ_3 \infty\succ_3 c\succ_3 b$};		\end{tikzpicture}
	\end{center}
	\caption{A 3-terminal example with three players $I=\{1,2,3\}$ and terminals $T=\{a,b,c\}$. Player $2$ controls two positions, while $1$ and $3$ controls one-one positions. The initial position $v_0$ is the one controlled by player $1$. The player's preferences are shown next to the figure.    This example has no NE.  This example is not a terminal DG game, since two of the players rank $\infty$ higher than two-two of the terminal outcomes, and it is not play-once either, since player $2$ controls two positions.
 \label{fig-exMartin-Vladimir}}
\end{figure}
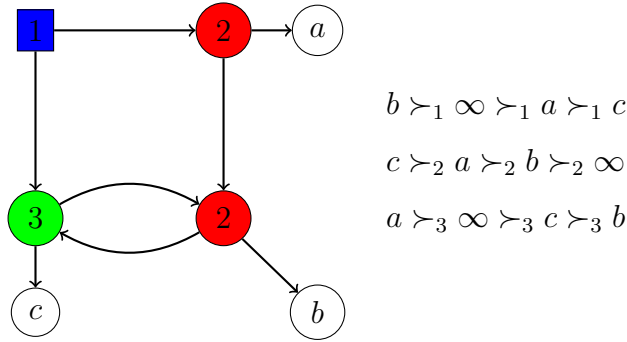


\section{Terminal DG games} 
\label{s4}


For $|T| \leq 2$ this result was obtained in \cite{BG03}. 
Due to Lemma \ref{l-notreachable}, we can assume 
for the rest of our proof that $T$ is reachable from $v_0$ in $\Gamma$.

Our remaining proof is inductive. We consider a potential ``smallest'' counter-example (smallest number of positions, and within that smallest number of moves) and derive a contradiction from its existence. To arrive to such a proof we need a few technical lemmas. 

We say that a move $(u,a)\in E$, $a\in T$  is a \emph{$\#k$-move} if terminal $a$ is the $k$th best outcome of player $i(u)\in I$. 
We show that the following properties hold for potential minimal counter-examples:

\begin{description}
	\item[(NoDE)] No dead end: $T$ is reachable from all positions $v\in V\setminus T$. 
	\item[(NoDummy)] No dummies: $d^+(v)\geq 2$ for all $v\in V\setminus T$. 
	\item[(NofT)] No forced termination: $N^+(v)\not\subseteq T$ holds for all $v\in V\setminus T$.
	\item[(No1M)] $\Gamma$ has no $\#1$-moves. 
	\item[(No3M)] $\Gamma$ has no $\#3$-moves. 
\end{description}

\begin{lemma}\label{l-no-deadend}
	If $\Gamma$ is a minimal counter-example, then the set of terminals is reachable from all positions $v\in V\setminus T$. In particular,  it satisfies condition (NoDE).
\end{lemma}
\begin{proof}
	Let us define $Q\subseteq V$ as the set of positions of $\Gamma$ from which $T$ cannot be reached, and assume for a contradiction that $Q\neq \emptyset$. Note that we have $v_0\in V\setminus Q$, since we assume that the set of terminals are reachable from the initial position. Then $\Gamma'=\Gamma[V\setminus Q]$ is a smaller game than $\Gamma$ terminal DG game.  Thus, it has a terminal NE $\sigma'=(\sigma'_i\in \Sigma_i(\Gamma')\mid i\in I)$ by the minimality of $\Gamma$. Let us then extend these strategies to $\Gamma$  by defining for all $v\in V_i$, $i\in I$ 
	\[
	\sigma_i(v) ~=~ 
	\begin{cases}
		\sigma'_i(v) & \text{ if } v\in V_i\setminus Q,\\
		w & \text{ for an arbitrary } w\in N^+(v) \text{ if } v\in V_i\cap Q.
	\end{cases}
	\]
	We claim that $\sigma$ is a NE of $\Gamma$, contradicting the assumption that $\Gamma$ is a counter-example to Theorem \ref{t-3T}. 
	
	To see this claim, let us note that $P(\sigma)=P(\sigma')$ and hence $g(\sigma)=g(\sigma')\in T$. Furthermore, a player $i\in I$ can improve on $\sigma$ only if he can deviate from $P(\sigma)$ to a position of $Q$, since all other positions have the same moves as in $\sigma'$, which is an equilibrium in $\Gamma'$. Such a deviation however yields a play that is not final by the definition of $Q$, and hence cannot be a strict improvement for player $i$ since all players prefer all terminals to $\infty$ in a terminal DG game.
\end{proof}

\begin{lemma}\label{l-no-dummies}
	If $\Gamma$ is a minimal counter-example, then it satisfies condition (NoDummy).
\end{lemma}
\begin{proof}
	By Lemma \ref{l-no-deadend} we have $d^+(v)\geq 1$ for all positions $v\in V\setminus T$. Assume for a contradiction that $d^+(v)=1$ for some $v\in V_i$, $i\in I$. (Such a $v$ is called a dummy position sometimes since the player controlling  it has no real choice in selecting an outgoing move.) Let us denote by $(v,w)\in E$ the unique move from $v$, set $F=\{(v,w)\}\cup \{(w,u)\mid u\in N^+(w)\}$, and define 
	\[
	V'=V\setminus \{w\} \text{ and } E'=\left(E\setminus F\right)\cup \{(v,u)\mid u\in N^+(w)\}.
	\]
	Let us also define $v_0'=v_0$ if $v_0\neq w$, and set $v_0'=v$ if $v_0=w$. 
	Let us then consider the game $\Gamma'=(G'=(V',E'),v_0',\prec)$. 
	
	Note that there is a one-to-one correspondence between the situations of $\Gamma$ and $\Gamma'$. Namely, to a situation $\sigma'$ of $\Gamma'$ let us associate $\sigma$ defined by $\sigma(v)=w$, $\sigma(w)=\sigma'(v)$, and $\sigma(u)=\sigma'(u)$ for all $u\not\in \{v,w\}$ (and vice versa.) 
	
	Since $\Gamma'$ is a smaller terminal DG game than $\Gamma$, it has a terminal NE $\sigma'$, because we assumed that $\Gamma$ is a minimal counter example. It is easy to see that the unique associated situation $\sigma$ of $\Gamma$ is a terminal NE of $\Gamma$. This contradicts the assumption that $\Gamma$ is a counter-example to Theorem \ref{t-3T}, and hence completes the proof of the lemma.
\end{proof}

\begin{lemma}\label{l-no-forced-terminal}
	If $\Gamma$ is a minimal counter-example, then it satisfies condition (NofT).
\end{lemma}
\begin{proof}
	Assume for a contradiction that for some position $v\in V_i\setminus T$, $i\in I$ we have $N^+(v)\subseteq T$. Let $a\in N^+(v)$ be the best outcome for player $i$ in the set $N^+(v)$, and define game $\Gamma'$ by removing moves $(v,b)$, $b\in N^+(v)$, $b\neq a$ from $\Gamma$. Since by Lemma \ref{l-no-dummies} we have $d^+(v)\geq 2$, game $\Gamma'$ is strictly smaller than $\Gamma$. Thus, it has a terminal NE $\gs'$. 
	
	We claim that $\gs'$ is a terminal NE in $\Gamma$, too. This leads to a contradiction, proving our statement. 
	
	To see the claim notice that all moves in $\Gamma'$ are also available in $\Gamma$, and hence $\gs'$ is a situation of $\Gamma$, too. Furthermore, the only moves that are available in $\Gamma$ and not available in $\Gamma'$ are the removed  $(v,b)$, $b\in N^+(v)$, $b\neq a$ moves. Thus, only player $i$ has more options in $\Gamma$ than in $\Gamma'$, and hence only he could improve on $\gs'$ in $\Gamma$.
	Furthermore,  player $i$ could improve on $\gs'$ only by utilizing one of the deleted moves. In that case however he could reach $v$ already in game $\Gamma'$, 
    and thus, by our choice of $a\in N^+(v)$, he cannot strictly improve on $\gs'$.
	\end{proof}

\begin{lemma}\label{l-no-1-moves}
	If $\Gamma$ is a minimal counter-example, then  it satisfies condition (No1M). 
\end{lemma}
\begin{proof}
	Assume for a contradiction that there exists a move $(u,a)$, $a\in T$ such that $a$ is player $i(u)$'s best outcome. Let us now consider the game $\Gamma'$ obtained from $\Gamma$ by deleting all the other moves from $u$. Note, that by Lemma \ref{l-no-dummies}, we had more than one move from $u$, thus, $\Gamma'$ is strictly smaller than $\Gamma$. Consequently, $\Gamma'$ has a terminal NE, $\gs'$, since we assumed that $\Gamma$ is a minimal counter example. We claim that $\gs'$ is a NE also in $\Gamma$, contradicting our assumption that it is a counter-example to Theorem \ref{t-3T}, and hence proving the statement.
	
	To see our claim, note that only player $i(u)$ has moves in $\Gamma$ that are not available in $\Gamma'$ (from position $u$), and that his \#1 move $(u,a)$ is available in both games. Thus, if he can deviate from the play $P(\gs')$ to $u$ in $\Gamma$, then he can do the same in $\Gamma'$, too, and thus, we must have $g(\gs')=a$ since $\gs'$ is a NE in $\Gamma'$, and hence he cannot improve on the outcome which is already his best possible one. 
	Or otherwise, he cannot reach $u$, in which case he cannot improve in $\Gamma$ either. 
\end{proof}

\begin{lemma}\label{l-no-3-moves}
	If $\Gamma$ is a minimal counter-example, then  it satisfies condition (No3M). 
\end{lemma}
\begin{proof}
	Assume for a contradiction that there exists a move $(u,c)$, $c\in T$, where $c$ is player $i(u)$'s worst terminal outcome (since we have $|T|=3$.) Let us consider game $\Gamma'$ obtained from $\Gamma$ by deleting move $(u,c)$. Then $\Gamma'$ is strictly smaller than $\Gamma$. Thus, $\Gamma'$ has a terminal NE, $\gs'$, since we assumed that $\Gamma$ is a minimal counter example. We claim that $\gs'$ is a NE in $\Gamma$, too, leading to a contradiction with the selection of $\Gamma$, and hence proving our statement. 
	
	To see our claim, note that $g(\gs')\in T$, and thus, we have $g(\gs')\succ_{i(u)} c$, since $c$ is the worst terminal outcome of player $i(u)$. Consequently, adding back the $(u,c)$ move cannot help player $i(u)$ to improve on $\gs'$ in $\Gamma$, either. 
\end{proof}

\begin{corollary}\label{cor-2moves}
	If $\Gamma$ is a minimal counter-example then  all terminal moves in it are $\#2$-moves. 
	\qed
\end{corollary}

\begin{lemma}\label{l-switching}
	Assume that $\Gamma$ satisfies conditions (NofT) and (No1M), and that $\gs\subseteq E$ is a terminal NE with $g(\gs)=a\in T$. Let us define a new situation $\gs'$ by defining for all $u\in V\setminus T$ 
	\[
	\gs'(u)~=~
	\begin{cases}
		v\in N^+(u)\setminus T & \text{ if } \gs(u)\in T\setminus \{a\},\\
		\gs(u) & \text{ otherwise}.
	\end{cases}
	\]
	Then $\gs'$ is also a terminal NE of $\Gamma$.
\end{lemma}
\begin{proof}
	Let us note first that $\gs'$ is well defined, due to condition (NofT).
	
	Assume for a contradiction  that a player $i\in I$ can improve on $\gs'$ by switching to situation $\gs''$. Then the play $P(\gs'',v_0)$ must pass through at least one of the positions $u\in V$ for which $\gs'(u)\neq \gs(u)$ since otherwise $\gs''$ would be an improvement over $\gs$, contradicting the assumption that $\gs$ is a NE. This implies that player $i$ could deviate from $\gs$ and reach terminal $\gs(u)\in T$. Since $\gs$ is a NE by our assumption, it follows that 
	\begin{equation}\label{e-gsu<a}
		\gs(u)\not\succ_i a.
	\end{equation}
	Since we assume that $\gs''$ is an improvement on $\gs'$ we must have 
	\begin{equation}\label{e-newa}
		g(\gs'',v_0)\succ_i a=g(\gs',v_0)=g(\gs,v_0).
	\end{equation}
	Since $\infty$ is not better than $a$ for any of the players,  and by the above two relations, it follows that $g(\gs'',v_0)$ is player $i$'s best outcome. Since in situation $\gs'$ terminal moves are only to $a\in T$, the last position before $T$ on $P(\gs'',v_0)$ must belong to $V_i$, in contradiction with condition (No1M). 
\end{proof}

\begin{proof}[\rm\textbf{Completion of the proof of Theorem \ref{t-3T}}]
	Assume for a contradiction that $\Gamma$ is a minimal counter-example. Since we assume that $T$ is reachable from $v_0$, it must have some terminal moves, and by Corollary \ref{cor-2moves} all such moves are $\#2$-moves. Let $(u,b)\in E(\Gamma)$ be such a $\#2$-move. By Lemma \ref{l-no-forced-terminal} condition (NofT) also holds, and thus, the game $\Gamma'$ defined by $E(\Gamma')=E(\Gamma)\setminus \{(u,v)\mid (u,v)\in E(\Gamma), ~ v\not\in T\}$ is strictly smaller than $\Gamma$, and hence it must have a terminal NE $\gs'\subseteq E(\Gamma')$, since $\Gamma$ is assumed to be a minimal counter-example. 
	\begin{description}
		\item[Case 1: $g(\gs')=b$.] In this case by Lemma \ref{l-switching} we have another terminal NE $\gs''\subseteq E(\Gamma')$ such that $g(\gs')=g(\gs'')=b$ and $\gs''$ has no moves to $T\setminus\{b\}$. We claim that $\gs''$ is also a terminal NE of $\Gamma$. 
		
		To see this claim, note that only player $i(u)$ could improve on $\gs''$ in $\Gamma$, and only if he could reach his \#1 terminal, since $g(\gs'')=b$ is already his second best outcome and he/she does not prefer $\infty$ to $b$. Since $\gs''$ has no moves to $T\setminus \{b\}$ and since $\Gamma$ satisfies property (No1M) by Lemma \ref{l-no-1-moves}, such an improvement is impossible.
		
		\item[Case 2: $g(\gs')\neq b$.] We claim that $\gs'$ is also a terminal NE of $\Gamma$. 
		
		To see this claim, note that only player $i(u)$ could improve on $\gs'$ in $\Gamma$, and only if he could reach position $u$ and utilize one of the deleted moves. In this case however, he can reach position $u$ already in $\Gamma'$ and could deviate to outcome $b$. Since $\gs'$ is a NE of $\Gamma'$, we must have $g(\gs')\succ_{i(u)} b$ due to strict preferences, implying that $g(\gs')$ is the best terminal outcome of player $i(u)$. Since condition $\infty$ is his/her worst outcome, he cannot improve on $\gs'$. 
	\end{description}
	
	In both cases we could show that $\Gamma$ has a terminal NE, contradicting our assumption that it is a counter-example. This contradiction proves our statement. 
\end{proof}

\begin{remark} $~$
	\begin{itemize}
		\item Note that the proof of Case 2 above depends on strict preferences. 
		\item Note that \cite{BGMOV18} has a 3-terminal example in which there is no NE and in which $\infty$ ranks better than two of the terminals for some of the players, see Figure \ref{fig-exMartin-Vladimir}.
		\item The next example in Figure \ref{fig-example2} shows that a 3-terminal game may not have a terminal NE if $\infty$ ranks better than one of the terminals (though it still has a non-terminal NE). This example can naturally be extended to any number of players.
	\end{itemize}
\end{remark}

\begin{figure}[ht]
	\begin{center}
		\begin{tikzpicture}[scale=0.8,fill opacity=0.5]
			\node[circle,draw,fill=blue,opacity=0.5] (v1) at (4,0) {$1$};
			\node[circle,draw,fill=red,opacity=0.5] (v3) at (8,0) {$3$};
			\node[circle,draw,fill=green,opacity=0.5] (v2) at (6,3.46) {$2$};
			
			\node[circle,draw,fill=white] (a) at (2,0) {$a$};
			\node[circle,draw,fill=white] (b) at (6,5.46) {$b$};
			\node[circle,draw,fill=white] (c) at (10,0) {$c$};
			\draw[thick,->] (v1) to (v2);
			\draw[thick,->] (v1) to (a);
			\draw[thick,->] (v2) to (v3);
			\draw[thick,->] (v2) to (b);
			\draw[thick,->] (v3) to (v1);
			\draw[thick,->] (v3) to (c);
			
			\node[right] () at (8,5.46) {$b\succ_1 c\succ_1\infty\succ_1 a$};
			\node[right] () at (8,4.46) {$c\succ_2 a\succ_2\infty\succ_2 b$};
			\node[right] () at (8,3.46) {$a\succ_3 b\succ_3\infty\succ_3 c$};
						
		\end{tikzpicture}
	\end{center}
	\caption{A 3-terminal play-once example with three players $I=\{1,2,3\}$ and terminals $T=\{a,b,c\}$. Each player controls one of the positions, as indicated on the picture. The initial position $v_0$ could be any of the three positions controlled by the players. The player's preferences are shown next to the figure.   This example is not a terminal DG game, since all players rank $\infty$ higher than one of the terminal outcomes. This example has no terminal NE, though the 3-cycle is a non-terminal NE.  
 \label{fig-example2}}
\end{figure}
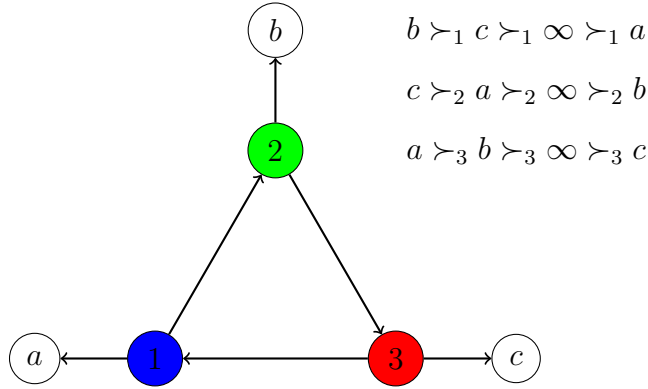

\section{Play-once DG games}
\label{s5} 

In this section we focus on terminal games for which $|V_i|=1$ for all players $i\in I$. We call such games \textit{play-once}. We prove that such games always have a NE, though some games may have only non-terminal NE, see Figure \ref{fig-example2}.

Let us observe first that without loss of generality 
we can assume that $N^+(v)\not\subseteq T$ and that 
\begin{equation}\label{e-at-most-one-termination}
|N^+(v)\cap T|\leq 1
\text{ for all positions } v\in V. 
\end{equation}
For positions with $|N^+(v)\cap T|=1$ we denote by $a(v)\in N^+(v)\cap T$ the terminal to which $v$ has a move. 

Let us define a partition of the positions. Let us define first
\begin{equation}\label{e-S}
    S=\{v\in V\mid N^+(v)\cap T\neq\emptyset \text{ and } a(v)\succ_{i(v)}\infty \}
\end{equation}

Let us next define 
\[
W=\{v\in V\mid v \text{ is reachable from } v_0 \text{ by a directed path in } G[V\setminus S]\}
\]
Let us next define $Q\subseteq V\setminus (S\cup W\cup T)$ as the set of vertices from which a directed cycle is reachable
in $G[V\setminus (S\cup W\cup T)]$ by a directed path (possible empty path). 
Finally we set $R=V\setminus (S\cup W\cup T\cup Q)$. Note that the set $R$ is acyclic (and could be empty). 

\begin{lemma}
\label{l-Wcycle}
A play-once game $G$ has a non-terminating NE if and only if $G[W]$, the graph induced by the set $W$, has a cycle. 
\end{lemma}

\begin{proof}
For the if part, 
 choose as a play such a cycle $C\subseteq W$ together with a path from $v_0$ to $C$. 
 For all the other vertices we choose non-terminating moves. 
It is not difficult to see that this produces a non-terminal NE. 

On the other hand, if $W$ is acyclic, then any play either terminates in $W$, or goes through a vertex of $S$. In the latter case it cannot be an infinite play by the definition of $S$, since each vertex $v \in S$ prefers its terminal move to the infinite play.
\end{proof}

Note that by the above definitions if $W\neq \emptyset$, then we have $v_0\in W$, while if $W=\emptyset$, then $v_0\in S$.

By the above lemma we can assume for the rest of our proof that $W$ induces an acyclic graph (or it is empty). 

\begin{lemma}\label{l-cycles-S}
If $G[W]$ is acyclic, then all directed cycles in $G[V\setminus Q]$ contain a vertex from $S$. 
\end{lemma}

\begin{proof}
    Immediate by the above definitions.
\end{proof}

\begin{lemma}\label{l-main-terminalNE}
    If $G[W]$ is acyclic, then there exists a terminal NE.
\end{lemma}

\begin{proof}
    Let us denote by $E'$ the set of moves of the induced subgraph $G[V\setminus Q]$ and by $D(S)=\{(u,v)\in E'\mid u\in S\}$
    the set of arcs leaving 
    vertices of $S$.
    
    Let us next consider a minimal subset $F\subseteq D(S)$ such that $G'=(V\setminus Q,E'\setminus F)$ is acyclic. By Lemma \ref{l-cycles-S} such a set exists.

    Let us next apply backward indiction on the acyclic graph $G'$, denote by $\sigma\subseteq E'\setminus F$ the obtained situation, by  $P(\sigma)$ the corresponding terminal play, and by $g(\sigma)\in T$ the terminal endpoint of $P(\sigma)$.  

    Let us consider in $G'$ a maximal directed tree $U$ rooted at $g(\sigma)$ that contains $P(\sigma)$, and define $\sigma'(v)=u$ for all moves $(v,u)\in U$. For vertices $w\in V\setminus Q$ that does not belong to $U$ we set $\sigma'(w)=\sigma(w)$. Finally, for vertices $w\in Q$ we define $\sigma'(w)$ such that $(w,\sigma'(w))$ is a non-terminating move from $w$.

    Note that $\sigma'$ is a situation in $G$ and we have $P(\sigma')=P(\sigma)$. 

    We claim that $\sigma'$ is a terminal NE in $G$. To see this claim let is first observe that for all positions $w\not\in Q$ we have 
    \begin{equation}\label{e-sigma'}
        \text{either } g(\sigma',w)=g(P(\sigma)) \text{ or } g(\sigma',w)=g(\sigma,w).
    \end{equation}

Since only positions along the play $P(\sigma')=P(\sigma)$ could improve on $\sigma'$, we consider a position $w$ along the play $P(\sigma')$ and a potential improving move $(w,u)\in E$. We break our proof into three cases, as follows:

\begin{itemize}
    \item[Case 1:] $(w,u)\in E'\setminus F$,  
    is an arbitrary move from $w$ not along the play. Since $G'$ is acyclic, one cannot create a cycle by switching $\sigma'(w)$ to $u$. Hence, by \eqref{e-sigma'} such switching 
    can earn player $w$ either $g(P(\sigma))$ or $g(\sigma,u)$. By the properties of backward induction we have $g(\sigma,u)\preceq_w g(\sigma,w)=g(P(\sigma))$ since the same switch is available for $w$ in $G'$. Consequently the switch to $u$ yields no improvement in both cases.
    \item[Case 2:] $(w,u)\in F$. By the definition of $F$ we must have $w\in S$. By the minimality of $F$ there exists a path $u\to w$ in $G'$, and hence also in $\sigma'$ by the definition of $\sigma'$. Thus, switching $\sigma'(w)$ to $u$ would create a cycle. Since $w\in S$, we have $g(\sigma')\succeq_w a(w)\succ_w\infty$, implying that such a switch is not improving for player $w$.
    \item[Case 3:] $(w,u)\in E\setminus E'$. By the definitions of the partition of the vertices of $G$, we must have $w\in S$ and $u\in Q$. Note that a switch of $\sigma'(w)$ to $u$ is not changing $\sigma'$ inside the set $Q$. Therefore,
    since $g(\sigma',u)=g(\sigma,u)=\infty$ and we have $g(\sigma')\succeq_w a(w)\succ_w \infty$, such a move cannot be improving for player $w$.
\end{itemize}
Since we could not improve by all possible switches of $\sigma'$ from a position along its play, $\sigma'$ is a NE.
\end{proof}

Note that Lemma \ref{l-main-terminalNE} is not an if and  only if claim, since we may have a terminal NE even when the induced subgraph $G[W]$ is not acyclic. 

\bigskip

\noindent{\textbf{Proof of Theorem \ref{t-playonce}}:}
The statement follows by Lemmas \ref{l-Wcycle} and \ref{l-main-terminalNE}.
\qed

\section{Discussions}\label{s6}

We conjecture that condition $|T| \leq 3$ can be waved in Theorem \ref{t-3T}. 
In other words, every terminal DG game has a NE and, moreover 
it has a terminal NE whenever 
a terminal can be reached by a directed path from the initial position. 

If the first part of the above conjecture fails, 
that is, there exists an NE-free terminal DG game, 
then there exists also a terminal DG game that has a non-terminal NE, 
but no terminal one \cite{BGMOV18}, 
although for any such game 
the non-terminal outcome $\infty$ is worse than 
all terminal outcomes for all players.  

\smallskip 

The following conjecture also remains open (Catch 22 \cite{Gur21}):    
Any NE-free DG game has at least 2 players for each of which  
the non-terminal outcome $\infty$ is better than at least 2 terminal outcomes, 
like in example of Figure \ref{fig-exMartin-Vladimir}. 

\smallskip

Instead of merging all non-terminal plays (lassos),  
let us assign a separate outcome to  all  plays  that end 
in a fixed strongly connected components of the graph \cite{Tar72,Sha81}.  
The obtained family of games was introduced in \cite{Gur18}, 
where they were called {\em multistage DG games}.
As we know, merging outcomes can create new NE but cannot destroy old ones. 
It is open if Theorems 1 and 2 can be extended to the multistage DG games. 
In particular, it is open if each terminal multistage DG game has a NE.
Conjecture Catch 22 was recently disproved for this family \cite{BGLNP24}. 
It was proven in \cite{Gur18} 
that every 2-person multistage DG game has a NE. 
The case when 
each directed cycle is a separate outcome was considered in \cite{BGMW11}. 
Nash-solvability was characterized in case of 
two players and symmetric directed graphs.

\smallskip

Even play-once and terminal DG games may have {\em improvement cycles};  
see \cite[Figure 5 and Section 2]{AGH10} 
for definitions, examples, and more details. 
Find sufficient and/or necessary conditions excluding improvement cycles. 
Note that any improvement cycle free DG game has a NE. 

\subsection*{Acknowledgements}
The second author worked within the framework
of the HSE University Basic Research Program.
The third author was partially supported 
by the joint project of Kyoto University and Toyota Motor Corporation, 
titled “Advanced Mathematical Science for Mobility Society”,  
and by JSPS KAKENHI, Grant Numbers JP20H05967 and JP19K22841. 


\begin{thebibliography}{99}

\bibitem{ARV09}
H. Ackermann, H. R{\"o}glin, and B. V{\"o}cking, 
Pure Nash equilibria in player-specific and weighted congestion games, Theoretical Computer Science, Internet and Network Economics 
410 (17) (2009) 1552--1563.

\bibitem{AGH10} 
D. Anderson, V. Gurvich, and T. Hansen, 
On acyclicity of games with cycles, 
Discrete Applied Math. 158:10 (2010) 1049--1063.

\bibitem{AHMS12}
D. Andersson, K. Hansen, P. Miltersen, and T. Sorensen, 
Deterministic graphical games, revisited, 
J. Logic and Computation 22:2 (2012) 165--178; 
Preliminary version in Fourth Conference 
on Computability in Europe (CiE-08), 
Lecture Notes in Computer Science 5028 (2008) 1--10.

\bibitem{BEGM12}
E. Boros, K. Elbassioni, V. Gurvich, and K. Makino,
On Nash equilibria and improvement cycles in
pure positional strategies 
for Chess-like and Backgammon-like n-person games,
Discrete Math. 312:4  (2012) 772--788.

\bibitem{BEGV24}
E. Boros, K. Elbassioni, V. Gurvich, and M. Vyalyi,
Two-person positive shortest path games have
Nash equlibria in pure stationary strategies, 
\url{https://arxiv.org/abs/2410.09257} (2024) 12 pp.  

\bibitem{BFGV23}
E. Boros, P. G. Franciosa, V. Gurvich, and M. Vyalyi, 
Deterministic $n$-person shortest path and
terminal games on symmetric digraphs have Nash
equilibria in pure stationary strategies, 
Int. J. Game Theory 53 (2024) 449--473.

\bibitem{BG03}
E. Boros and V. Gurvich, 
On Nash solvability in pure stationary strategies of
positional games with perfect information which may have cycles,
Mathematical Social Sciences 46 (2003) 207--241.

\bibitem{BG09}
E. Boros and V. Gurvich, 
Why Chess and Backgammon can be solved 
in pure positional uniformly optimal strategies,
RUTCOR Research Report, RRR-21-2009, Rutgers University, 
https://rutcor.rutgers.edu/~gurvich/BackGammon.pdf 

\bibitem{BGMOV18}
E. Boros, V. Gurvich, M. Milanic, V. Oudalov, and J. Vicic,
A three-person deterministic graphical game without Nash equilibria,
Discrete Applied Math. 243 (2018) 21--38.

\bibitem{BGMW11}
E. Boros, V. Gurvich, K. Makino, and S. Wei,
Nash-solvabile two-person symmetric cycle game forms,
Discrete Applied Math. 159:15 (2011) 1461--1487.

\bibitem{BGOR14}
On Nash-solvability of chess-like games,
E. Boros, V. Gurvich, V. Oudalov,  R. Rand, 
RUTCOR Research Report, RRR-09-2014, Rutgers University, 
\url{https://www.semanticscholar.org/paper/On-Nash-solvability-of-chess-like-games-Boros-Gurvich/3f64a8d5adcc66015a84bd9d5971ede555e04703}.

\bibitem{BGLNP24}
Butyrin, V. Gurvich, A. Lutsenko, M. Naumova, and M. Peskin.
A counter-example to conjecture ``Catch 22",  
\url{https://arxiv.org/abs/2406.14587} (2024) 12 pp.

\bibitem{Con92}
A. Condon, The complexity of stochastic games, 
Information and Computation 96 (1992) 203--224.

\bibitem{DS02}
V. I. Danilov and A. I. Sotskov, Social Choice Mechanisms, 
in “Studies of Economic Design”, Springer, Berlin-Heidelberg, 2002.

\bibitem{EF70}
J. Edmonds and D.R. Fulkerson,
Bottleneck extrema, 
J. Combin. Theory  8:3 (1970) 299--306.

\bibitem{Gal53}
D. Gale, A theory of $N$-person games with perfect information,
Proc. Natl. Acad. Sci. 39 (1953) 496--501.

\bibitem{Gur17}
V. Gurvich,
Generalizing Gale's theorem on backward induction and domination of strategies, 
arXiv \url{http://arxiv.org/abs/1711.11353} (2017) 12 pp. 


\bibitem{Gur73}
V. Gurvich, To theory of multi-step games, 
USSR Comput. Math. and Math Phys. 13:6 (1973) 143--161. 

\bibitem{Gur75}
V. Gurvich, 
Solution of positional games in pure strategies, 
USSR Comput. Math. and Math. Phys. 15:2 (1975) 74--87.


\bibitem{Gur89} 
V. Gurvich, 
Equilibrium in pure strategies,
Soviet Math. Dokl. 38:3 (1989) 597--602. 


\bibitem{Gur17}
V. Gurvich, 
Generalizing Gale's theorem on backward induction 
and domination of strategies (2017) 16 pp. 
\url{http://arxiv.org/abs/1711.11353}.

\bibitem{Gur18}
V. Gurvich, Backward induction in presence of cycles,
Oxford Journal of Logic and Computation 28:7 (2018) 1635--1646.

\bibitem{Gur21}
V. Gurvich,
On Nash-solvability of finite $n$-person deterministic
graphical games, Catch 22 
\url{https://arxiv.org/abs/2111.06278} 
(2021) 4 pp. 

\bibitem{Gur21a}
V. Gurvich, 
On Nash-solvability of finite n-person shortest path games, 
bi-shortest path conjecture,  
\url{http://arxiv.org/abs/2111.07177} (2021) 5 pp.  

\bibitem{GK18}
V. Gurvich and G. Koshevoy, 
Monotone bargaining is Nash-solvable, 
Discrete Appl. Math. 250 (2018) 1--15.

\bibitem{GN21}
V. Gurvich and M. Naumova, 
Lexicographically maximal edges of dual hypergraphs and 
Nash-solvability of tight game forms, 
Annals of Mathematics and Artificial Intelligence, 92 (1) (2024) 49--57.

\bibitem{GN21A}
V. Gurvich and M. Naumova,
On Nash-solvability of n-person graphical games 
under Markov's and a-priori realizations, 
Annals of Operations Research, 336:3 (2024) 1905--1927. 

\bibitem{GN21B}
V. Gurvich and M. Naumova,
Computing lexicographically safe Nash equilibria 
in finite two-person games with tight game forms given by oracles,
Discr. Appl. Math. 340:3 (2023) 53--68.

\bibitem{GN22}
V. Gurvich and M. Naumova, 
On Nash-solvability of finite two-person tight vector game forms,  
\url{https://arxiv.org/abs/2204.10241} (2022) 17 pp.

\bibitem{GO14}
V. Gurvich and V. Oudalov,
A four-person chess-like game
without Nash equilibria in pure stationary strategies 
\url{https://arxiv.org/abs/1411.0349}, (2014) 17 pp., 
Business Informatics 1:31 (2015) 68--76. 

\bibitem{HKM21}
T. Harks,  M. Klimm, and J. Matuschke, 
Pure Nash equilibria in resource graph games,  
Journal of Artificial Intelligence Research (2021).  

\bibitem{HR04}
C. A. Holt and A. E. Roth, 
The Nash equilibrium: A 
National Academy of Sciences 101 (2004) 3999--4002.

\bibitem{HL97}
R. Holzman and N. Law-Yone, 
Strong equilibrium in congestion games, 
Games and Economic Behavior, 21 (1) (1997) 85--101.

\bibitem{Kuh53}
H. Kuhn, Extensive games and the problem of information,
in Contributions to the theory of games, 
Volume 2, Princeton (1953) 193--216.

\bibitem{KMMM90}
N. S. Kukushkin, I. S. Men'shikov, O. R. Men'shikova, and V.V. Morozov, 
Resource allocation games, 
Computational Mathematics and Modeling, 1 (4) (1990) 433.

\bibitem{Mil96}
I. Milchtaich, Congestion models of competition, 
The American Naturalist, 147 (5) (1996) 760--783.

\bibitem{Mil96_2}
I. Milchtaich, Congestion games with player-specific payoff 

\bibitem{MS96} D. Monderer and L. S. Shapley, 
Potential games, Games and Economic Behavior, 14 (1) (1996) 124--143.

\bibitem{Nas50}
J. F. Nash Jr., 
Equilibrium points in n-person games, 
Proceedings of the National Academy of Sciences 36:1 (1950) 48--49.

\bibitem{Nas51}
J. F. Nash Jr., Non-cooperative Games, 
Annals of Math. 54:2 (1951) 286--295.

\bibitem{Ros73}
R. W. Rosenthal, A class of games possessing pure-strategy Nash equilibria, International Journal of Game Theory, 2 (1973) 65--67. 

\bibitem{Sha81} M. Sharir, 
A strong-connectivity algorithm and its application in data flow analysis,
Comput. Math. Appl. 7 (1981) 67--72.

\bibitem{Tar72} 
R. E. Tarjan, 
Depth-first search and linear graph algorithms, 
SIAM J. Computing 1:2 (1972) 146–160.


\bibitem{Was90}
A. R. Washburn,
Deterministic graphical games,
J. Math. Analysis and Applications 
153 (1990) 84–-96.

\end{thebibliography}
\end{document}